\documentclass[runningheads]{llncs}
\usepackage{url}
\usepackage{theorem}
\usepackage{latexsym}
\usepackage{xspace}
\usepackage{amsmath}
\usepackage{amssymb}
\usepackage{fancybox}
\usepackage{epsfig}
\usepackage{graphicx}

\usepackage{url}

\usepackage{soul,xcolor}

\def\eps{\epsilon}

\newcommand{\hide}[1]{}
\hide{
\theoremheaderfont{\scshape}
\theorembodyfont{\slshape}
\newtheorem{theorem}{Theorem}
\newtheorem{definition}{Definition}
\newtheorem{lemma}[theorem]{Lemma}

\newtheorem{corollary}[theorem]{Corollary}

\theoremstyle{plain}
\theorembodyfont{\rmfamily}

\theorembodyfont{\itshape}

\newenvironment{proof}{\noindent {\sc Proof.  }}{$\Box$ \medskip}
\newenvironment{proofnodot}{\noindent {\sc Proof  }}{$\Box$ \medskip}

 {
    \begin{enumerate}}{\end{enumerate}}

\setlength{\oddsidemargin}{0pt}
\setlength{\evensidemargin}{0pt}
\setlength{\textwidth}{6.0in}
\setlength{\topmargin}{0in}
\setlength{\textheight}{8.5in}
}

\newcommand{\marginlabel}[1]%
{\mbox{}\marginpar{\it{\raggedleft\hspace{0pt}#1}}}
\newlength{\pgmtab}  %  \pgmtab is the width of each tab in the
\def\qed{ \ \vrule width.2cm height.2cm depth0cm\smallskip}
\def\eps{\epsilon}

\newcommand{\remove}[1]{}

\urldef{\mailsa}\path{cole@cs.nyu.edu}
\urldef{\mailsb}\path{karloff@yahoo-inc.com}
\begin{document}

\mainmatter
%\begin{frontmatter}

\title{Fast Algorithms for Constructing Maximum Entropy Summary Trees}
\titlerunning{Maximum Entropy Summary Trees}
\author{Richard Cole\thanks{
Computer science department, Courant Institute, NYU;
{\tt cole@cs.nyu.edu}.
%\thanks{cole@cs.nyu.edu. Richard Cole's work was supported in part by NSF grant CCF-1217989.
This research was supported in part by NSF grant CCF-1217989.}
\and
Howard Karloff\thanks{Yahoo Labs, New York, NY;  karloff@yahoo-inc.com.}
%\institute{ Computer Science Department, Courant Institute, NYU.  \and Yahoo Labs.  }
\institute{$ $}
}

%\date{$ $}
\maketitle

\sloppy

\begin{abstract}
Karloff and Shirley recently proposed ``summary trees" as a new
way to visualize large rooted trees (Eurovis 2013) and gave
algorithms for generating a maximum-entropy $k$-node summary tree
of an input
$n$-node rooted tree.    However, the algorithm generating
optimal summary trees was only pseudo-polynomial (and worked only
for integral weights);  the authors left open existence of a
polynomial-time algorithm.
In addition,
the authors provided an additive approximation algorithm and a
greedy heuristic, both working on real weights.

This paper shows how to construct maximum entropy $k$-node summary trees
in time $O(k^2 n +n\log n)$ for {\em real} weights (indeed, as small as
the time bound for the greedy heuristic given previously);
how to speed up the
approximation algorithm so that it runs in time
$O(n + (k^4/\epsilon) \log (k/\epsilon))$,
and how to speed up
the greedy algorithm so as to run in time $O(kn+n \log n)$.
Altogether, these results make summary trees a much more
practical tool than before.
\end{abstract}

% \end{frontmatter}

%\newpage
%\setcounter{page}{1}

\section{Introduction}

How should one draw a large $n$-node rooted tree on a small sheet
of paper or computer screen?  Recently, in Eurovis 2013, Karloff and Shirley
\cite{KS} %\footnote{best paper honorable mention.}
proposed a new way to visualize large trees.  While the best
introduction to summary trees appears in \cite{KS}, here we give a
necessarily short description.    A user has an $n$-node
node-weighted tree $T$ and wants to draw a $k$-node summary $S$ of $T$
on a small screen or sheet of paper, $k$ being user-specified.
We begin with an informal, bottom-up, operational description.
Two types of contraction are performed:
subtrees are contracted to single nodes that represent the corresponding subtrees;
similarly multiple sibling subtrees (subtrees whose roots are siblings)
are contracted to single nodes representing them.
The node resulting from the latter contraction is called a group node.
The one constraint is that each node in the summary tree
have at most one child that is a group node.
Examples are shown in Figure~\ref{fig:definition}--\ref{fig:websummary}
below (these figures appeared originally in~\cite{KS}).
\begin{figure}[htb]
\centering
\includegraphics[width=\linewidth]{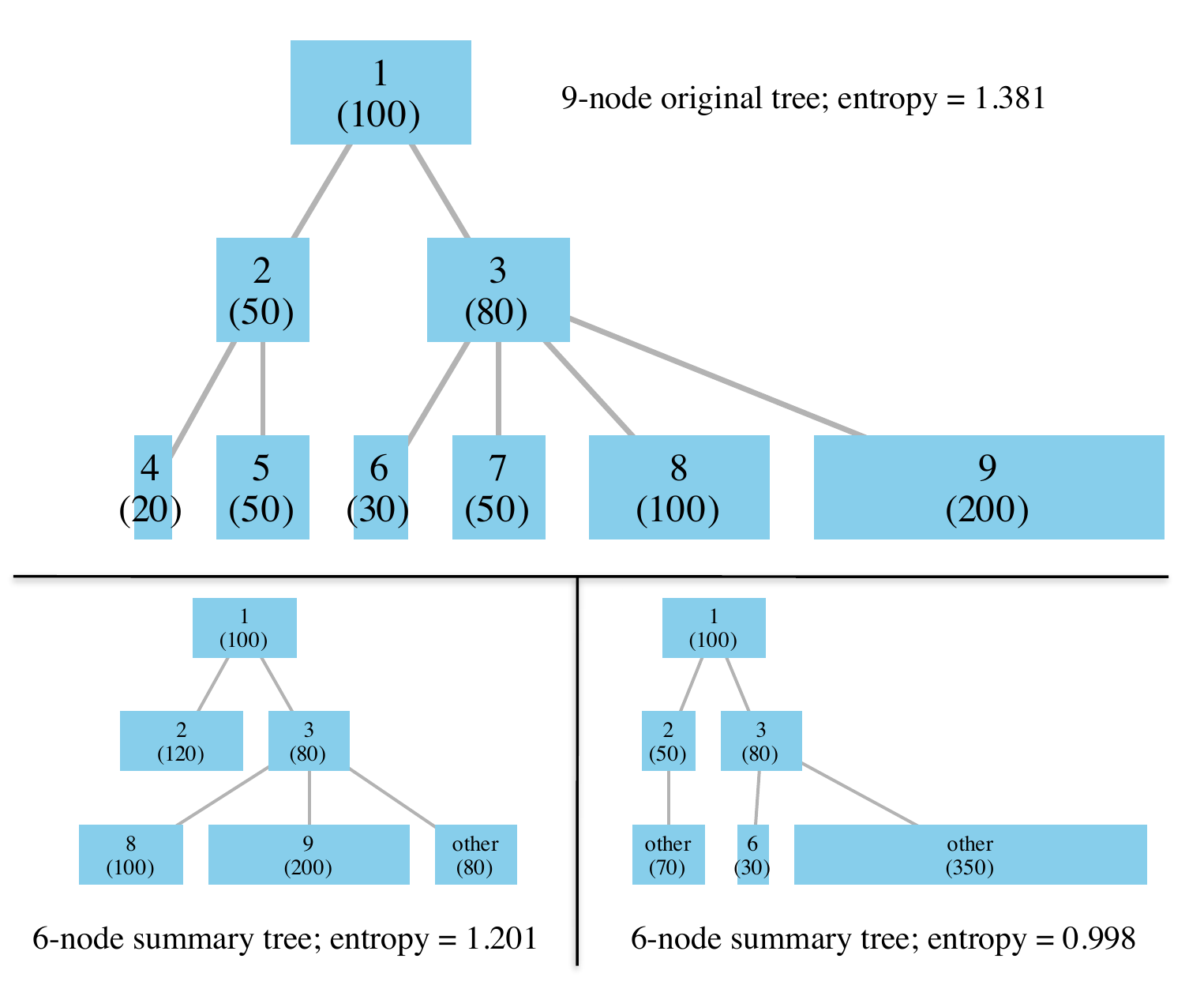}
\caption{\label{fig:definition} In the upper panel, a 9-node tree
(with node weights in
parentheses), and below it, two different 6-node summary trees of
the original 9-node tree.}
\end{figure}

\begin{figure*}[htb]
  \centering
  \mbox{} \hfill
  % the following command controls the width of the embedded PS file
  % (relative to the width of the current column)
  \includegraphics[width=\linewidth]{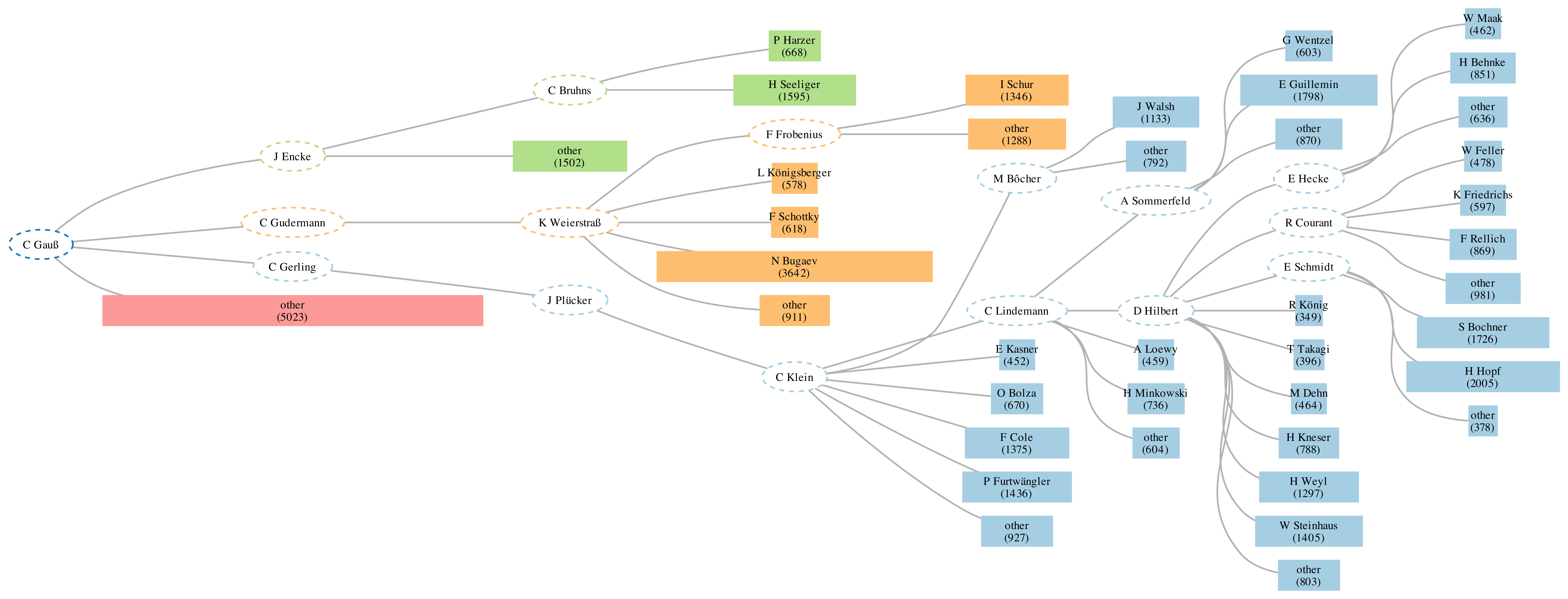}
  % replacing the above command with the one below will explicitly set
  % the bounding box of the PS figure to the rectangle (xl,yl),(xh,yh).
  % It will also prevent LaTeX from reading the PS file to determine
  % the bounding box (i.e., it will speed up the compilation process)
  % \includegraphics[width=.3\linewidth, bb=39 696 126 756]{sampleFig}
  \hfill
  \hfill \mbox{}
  \caption{\label{fig:teaser}
  Taken from \cite{KS}, the maximum entropy 56-node summary tree of the math genealogy
tree rooted at Carl Friedrich Gauss, which has 43,527
equal-weighted nodes (where the original advisor-student graph was
forced to be a tree by choosing the primary advisor for
each student who had multiple advisors). Node colors are
determined by their depth-1 ancestor, and node areas are
proportional to their weights in the summary tree.  This tree is
best viewed (and enlarged) on a computer screen.}
\end{figure*}

Next, we give a more formal description.
Let $T_{v}$ denote the subtree of $T$ rooted at $v$.
We name each node of $S$ by the set of nodes of $T$ that it represents.
The following comprise the possible {\em summary trees for $T_{v}$}:
If $T_{v}$ has just one node, the only summary tree is the one node $\{v\}$.
Otherwise, a summary tree for $T_{v}$ is one of:
\begin{enumerate}
\item
a one-node tree $V(T_{v})$ (the set of nodes in $T_{v}$); or
\item
a singleton node $\{v\}$ and summary trees
for the subtrees rooted at the children of $v$ (and edges
from $\{v\}$ to the roots of these summary trees); or
\item
a singleton node $\{v\}$,
a node $other_v$ representing a non-empty subset $U_v$
of $v$'s children and
 all the descendants of the nodes $x \in U_v$,
and for each of $v$'s children $x\not \in U_v$
a summary tree for $T_{x}$ (and edges
from $\{v\}$ to $other_v$ and to the roots of the summary
trees for each $T_{x}$).%
\footnote{$other_v$ sets of size 1 are covered by Cases 2 and 3, but this
redundancy is convenient for the algorithm description.}
Sometimes we will overload the term $other_v$ by using it to denote the subset $U_v$.
\end{enumerate}
We allow arbitrary nonnegative real weights $w_v$ on the nodes $v$ of the input tree $T$.
The weight of a node in a summary tree is defined to be the
sum of the weights of the corresponding nodes in $T$.
Paper~\cite{KS} defined the entropy of a $k$-node summary
tree with nodes of weights $W_1,W_2,...,W_k$ to be
$-\sum_{i=1}^k p_i \lg p_i$, where $p_i=W_i/W$ and $W$ is
the sum of all node weights, the usual
information-theoretic entropy.
Paper~\cite{KS} then
proposed that the most informative summary
trees are those of maximum entropy.
As noted in~\cite{KS}, this is a natural way
to think about the information contained in a node-weighted
tree.  For given a bound on the number of nodes available in
a summary tree, it seems plausible that a best summary tree is one of maximum entropy,
because it is theoretically the most informative.
This provided a principled way to identify the best $k$-node summary tree,
in contrast to more heuristic and operational rules in prior work.

The fact that $other_v$ is an {\em arbitrary} non-empty
subset of $v$'s potentially large set of children is what makes
finding maximum entropy summary trees difficult.  Indeed,
\cite{KS} resorted to using a dynamic program over the node
weights (which worked provided that the weights were
integral) and which led to a final running time of $O(K^2n W)$,
where $W$ is the sum of the node weights and $K$ is the maximum
$k$ for which one is interested in finding a $k$-node summary
tree.  Given $K$, the dynamic program finds maximum entropy
$k$-node summary trees for $k=1,2,\ldots,K$;
from now on we assume that the user specifies $K$ and
$k$-node summary trees are found for all $k\le K$.    The
algorithm
worked well when $W$ was small, but failed to terminate on
two of the five data sets used in \cite{KS}.

The key to obtaining a running time independent of $W$ is to develop
a fuller understanding of the structure of maximum entropy summary trees.
Our new understanding readily yields a truly polynomial-time algorithm.
The main remaining challenge is to create and analyze an effective implementation.
We give an algorithm
running in time $O(K^2 n+n\log n)$~\footnote{Actually,
this can be reduced to $O(K^2 n)$ time by using a combination of fast selection
and sorting instead of sorting alone in various places.};
it generates
maximum entropy summary trees
even for real weights, assuming, of course, a real-arithmetic
model of computation, which is necessary (even for integral
weights) because of the
computation of logarithms.
%(In the usual
%binary model of computation, even determining if a sum of the
%form $\sum_i w_i \lg w_i \le B$ or not, when $B$ and the $w_i$'s
%are large integers, is not known to be in NP;  see
%\cite{yannakakis}.)
This result is based on a structural
theorem which shows that the $other$ sets, while allowed to be
arbitrary, can be assumed, without loss of generality, to have a simple structure.

\begin{table}[t]
\centering
\begin{tabular}{| c| |c | c| c|}
\hline
 & Optimal Entropy & Greedy & $\eps$-Approximate \\ [0.5ex]
\hline
Known results & $O(K^2 n W)$~\cite{KS} & $O(K^2 n + n \log n)$~\cite{KS}  & $O(K^2n W_0)$~\cite{KS} \\
\hline
New results & $O(K^2 n+n\log n)$ & $O(Kn + n\log n)$  & $O(n + K^3 W_0 + W_0 \log W_0)$ \\
\hline
\end{tabular}
\vspace*{0.1in}
\caption{Running times of the algorithms; $W_0 = O((K/\eps) \log (K/\eps))$.}
\label{tbl:results}
\end{table}

\begin{figure}[htb]
\centering
\includegraphics[width=\linewidth]{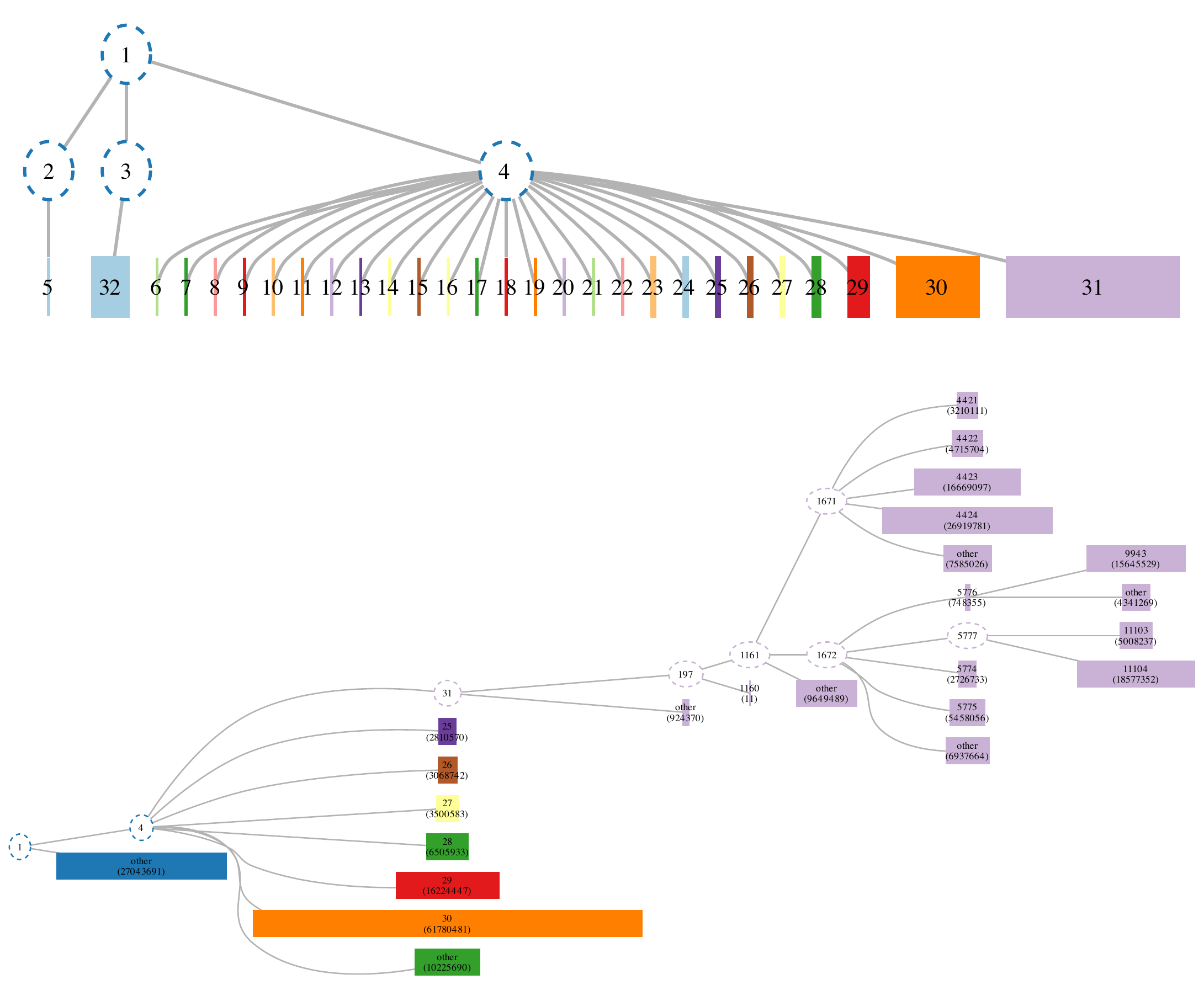}
\caption{\label{fig:websummary} Two summary trees of a
19,335-node web traffic tree. The upper figure is a naive
aggregation to depth 2; the node weights are heavily skewed.
The bottom figure is the maximum entropy $32$-node summary tree,
which displays much more information given the same number of
nodes.}
\end{figure}

To deal with the case of real weights or exceedingly large
integral weights, \cite{KS} gave an algorithm based on scaling,
rounding, and  algorithmic discrepancy theory which builds a
summary tree
whose entropy is within $\epsilon$ {\em additively} of the maximum,
in time $O(K^2n W_0)$, where $W_0 $
is $O((K/\epsilon)\log (K/\epsilon))$.
Keep in mind
here that $K$ is meant to be small, e.g., 100 or 500, while $n$
is meant to go to infinity, and also that $W_0$ is a function
only of $K$ and $\epsilon$ (and of neither $n$ nor $W$).   The key here was to show
that scaling the real input weights to have
sum $W_0$, rounding them using algorithmic discrepancy
theory, and then running the exact dynamic program previously
mentioned on the rounded weights caused a loss of only $\epsilon$ in the final entropy.

This paper shows that the same algorithm can be implemented in
time $O(n+K^3 W_0 + W_0 \log W_0)$;
this is linear time
if $n$ is larger than the other terms.
The key here is to notice that if the sum of integral weights
is $W_0$, which is small, and $n \gg W_0$, then most nodes have
rounded weight 0.  Surely one shouldn't have to devote a lot of
time to nodes of weight 0, and our algorithm, by effectively
replacing $n$ by $O(W_0)$,  exploits this intuition.

Last, \cite{KS} proposed a fast greedy algorithm to generate
summary trees.  Running in time $O(K^2 n+n\log n)$ (though
\cite{KS} overlooked the $n \log n$ time needed for sorting), the
algorithm never took longer than six seconds to run on the data
sets of~\cite{KS}.
This paper shows that a simple modification
to the greedy code,
neither suggested  in~\cite{KS} nor implemented in the associated C
code,
specifically,
not computing a $k$-node summary tree of a tree rooted at a node
having fewer than $k$ descendants, decreases the running time bound
of the greedy algorithm from $O(K^2n + n\log n)$ to $O(Kn+n\log
n)$.  While the modification is trivial, its analysis is not.

Taken together, these new results show that maximum entropy
summary trees are a much more practical tool than was previously
known.

\paragraph{Roadmap.}
Section \ref{background} describes earlier work on visualizing trees.
In Section~\ref{sec:structure} we prove the structural theorem on which our improved  algorithms depend.
This is followed in Section~\ref{sec:exact} with our exact algorithm
and in Section~\ref{sec:key-lemma} with the key lemma for analyzing the exact and greedy algorithms.
Section~\ref{sec:greedy} gives the greedy algorithm
and Section~\ref{sec:approx}, the approximate algorithm and its analysis.

%A number of the proofs are omitted for lack of space.

\section{Previous Work}\label{background}

\hide{
\noindent
{\bf Previous work.} Traditionally tree
visualization involved either visualizing the entire tree or
allowing the user to interactively specify tree parts of
interest. Approaches taken include ``Degree-of-interest trees"
\cite{card,heer}, ``hyperbolic browsers'' \cite{lamping}, and the ``accordion drawing
technique" \cite{beerman,munzner}.
``Space-filling'' layouts, e.g., treemaps \cite{shneiderman}, are another popular method.
Paper \cite{vL} is
a recent survey on techniques for drawing large graphs.
Also see \cite{KS} for other relevant previous work.
}

Traditionally tree visualization involved either visualizing the entire tree
or allowing the user to interactively specify in what part of the tree he or she is interested.
Obviously, if one draws a huge tree on a sheet of paper or a computer screen, not only will labels be close-to-impossible
to read, there will be too much information, in that the reader will not know on what part to focus.

Many researchers have attempted to ameliorate the issues involved with drawing a huge tree by allowing
interactivity.  Initially perhaps only the root of the tree is displayed.  When the user clicks on a node,
that node's children then appear.   ``Degree-of-interest trees''~\cite{card,heer} let a user explore a tree
interactively.
Other interactive techniques are ``hyperbolic browsers''~\cite{lamping} and
the ``accordion drawing technique''~\cite{munzner,beerman}.

Researchers have proposed ``space-filling''
layouts as an alternative to traditional node-and-edge layouts.
Treemaps~\cite{shneiderman} are one popular way to lay out large trees.  The root node is represented by
a rectangle, and recursively
the children of a node $v$ are represented by rectangles which together
partition the rectangle representing $v$.
But treemaps are not effective at showing the hierarchy of a tree.

Von Landesberger et al. wrote a recent survey~\cite{vL} on techniques for drawing large graphs.
Other relevant previous work can be found in~\cite{KS}.

\section{Structural Theorem}
\label{sec:structure}
This section proves a structural theorem which implies that
maximum entropy summary trees can be computed in polynomial
time,
in a real-arithmetic model of computation.
We begin by relating our approach to the greedy algorithm from~\cite{KS}.
Let $v$ be a node of
an
input tree and suppose that $\{v\}$ appears in the summary tree.
Recall that $other_v$ denotes the group child of $v$, if any.
\hide{
its one child of the form
$other$, if present.  (Recall that a node $\{v\}$ may have at most one of its children
of the form  $other_v$.  Each remaining child
is either of the form $\{x\}$ for a child $x$ of $v$ in the
input tree or is of the form $T_{x}$ for a child $x$ of $v$ and
represents $V(T_{x})$.
If $other_v$ is present, then $other_v$ represents a
nonempty subset of $v$'s children in $T$ together with their
descendants, i.e.,\ $other_v$
represents the union of $V(T_{x})$ over that subset of
$v$'s children.)
}

\begin{definition}
\begin{enumerate}
\item
The {\em size} $s_v$ of a node $v$ in $T$ is the sum of the weights of its
descendants.
\item
$n_{v}$ denotes the number of descendants of $v$ (including
$v$).
\item
$d_v$ denotes the degree of $v$, the number of children it has.
\item
$\langle v_1,v_2,...,v_{d_v}\rangle $ denotes the children of $v$ when sorted
into nondecreasing order by size.
(Fix one sorted order for each $v$, breaking ties arbitrarily.)
\item
The {\em prefixes}  of $\langle v_1,v_2,...,v_{d_v}\rangle$
are the sequences $\langle v_1,v_2,...,v_i\rangle$ and
sets $\{v_1,v_2,...,v_i\}$ for $i\ge 0$.
\end{enumerate}
\end{definition}

The greedy algorithm in~\cite{KS} sorted and then processed the children of each node in nondecreasing
order by size;
more about this later.
It finds a maximum entropy summary
tree among those
in which for each $v$, either $other_v$ does not exist or is a
nonempty prefix of $\langle
v_1,v_2,...,v_{d_v}\rangle$,
but this need not be the optimal summary tree.
In fact, \cite{KS} gives a 7-node tree $T$
for which the uniquely optimal 4-node summary tree has an $other_v$ node which is not
a prefix of $v$'s children (see Figure~\ref{7node}).
In their example,
the greedy algorithm achieves approximately 1 bit of
entropy, but the optimal summary tree achieves
approximately 1.5
bits
(and $1.5/1.0$ is the worst ratio between greedy and optimum of which we are aware).
This example proves that restricting
$other_v$ to be a prefix of the list of $v$'s children
can lead to summary trees of
suboptimal entropy.
Consequently,
% Because prefixes are not general enough,
\cite{KS} resorted to a pseudo-polynomial-time dynamic
program in order to find the optimal $other$ sets.

\begin{figure}[htb]
\centering
\includegraphics[width=\linewidth]{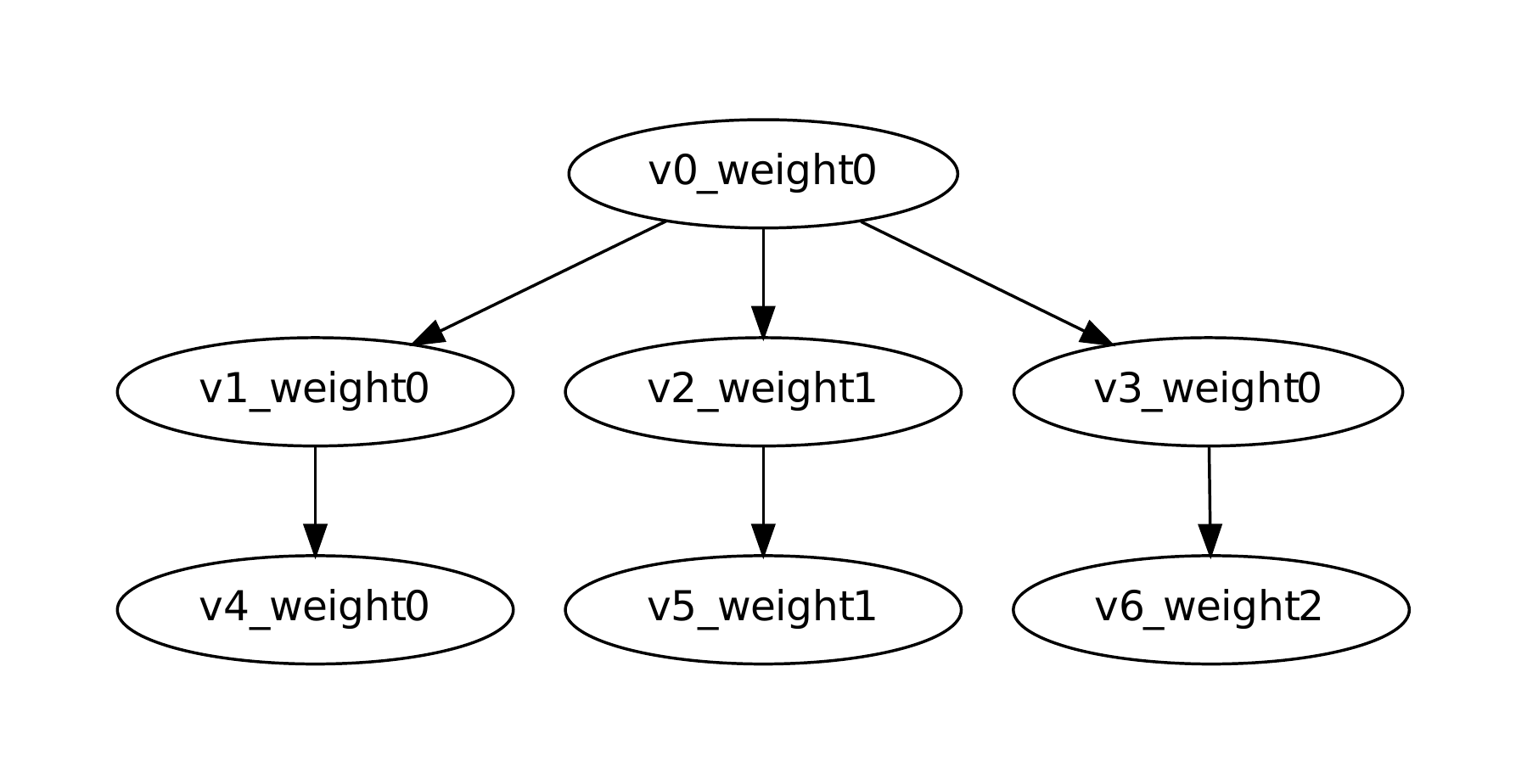}
\caption{\label{7node} A 7-node tree on which the
greedy algorithm does particularly badly.  Imagine that the
weight of node $v_6$ slightly exceeds 2, so that the unique
sorted order of the root's children into nondecreasing order by
size is $\langle v_1,v_2,v_3\rangle$.
The unique $4$-node maximum entropy summary tree has
$other_{v_0}=\{v_1,v_3\}$, which is not a prefix of $\langle
v_1,v_2,v_3\rangle$; this summary tree has entropy 1.5.  By
contrast, greedy gets $other_{v_0}=\{v_1,v_2\}$, in a summary
tree of entropy 1.}
\end{figure}

The definition of summary trees
allows $other_v$ to represent an {\em arbitrary} nonempty subset of
$v$'s children (and all their descendants).  However,
in this paper we prove the surprising fact that, without
loss of generality, in every summary tree {\em of maximum
entropy}, $other_v$ can be assumed to have a special form,
a simple extension of the ``prefix'' form used in the greedy algorithm from~\cite{KS}.

\begin{definition}
The {\em near-prefixes} of $\langle v_1,v_2,\ldots,v_{d_v}\rangle$ are
the sequences $\langle v_1,v_2,\ldots,v_i;v_j\rangle$ and the sets
$\{v_1,v_2,\ldots,v_i;v_j\}$
where $i\ge 0$, $j\ge i+2$, and $j\le d_v$.
$v_j$ is called the non-prefix element.
This terminology is also applied to the sequence
$\langle T_{v_1},T_{v_2},\ldots,T_{v_{d_v}}\rangle$ of trees rooted
at $ v_1,v_2,\ldots,v_{d_v}$, respectively.
\end{definition}

% The earlier paper proves that a $k$-node summary tree exists for $1\le k\le n$.
We prove the following structural theorem:
\begin{theorem}\label{structural}
For each $k$, $1 \le k \le n$,
there is a maximum entropy $k$-node summary tree $S$ in which, for every node $v$, $other_v$, when
present, is either a prefix or a near-prefix of $\langle
T_{v_1},T_{v_2},\ldots,T_{v_{d_v}}\rangle$.
\end{theorem}
\proof
For any summary tree $R$ of an $n$-node tree $T$,
let $M=2n+1$ and define $\Phi(R)=\sum_{v:other_v\mbox{
exists}} M^{n-d_R(v)}\sum_{j: v_j \in other_v} j$,
where $d_R(v)$ denotes the depth in $R$ of the node $other_v$.
Among all maximum entropy summary trees for $T$, let $S$ be
one for which $\Phi(S)$ is minimum.
(The role of $\Phi$ will be to enable tie-breaking among equal-weight summary trees.)

%In what follows, we simplify the calculations by rescaling the entropy to be computed using
%logs base 2 rather than base e.

\begin{lemma}
\label{lem:subtrees-as-single-nodes}
Let $v$ be a node of $T$ such that $other_v$ exists in $S$.
If $v_i \notin other_v$ and $v_j \in other_v$, where $i < j$,
then $T_{v_i}$ is represented by two or more nodes in $S$.
\end{lemma}
\proof
Suppose, for a contradiction, that $T_{v_i}$ is represented by a single node.
Consider the following alternate summary tree $S'$: $S'$ is obtained
from $S$ by replacing $v_j$ in $other_v$ by $v_i$,
and by representing
$T_{v_j}$ by a single node.  The number of nodes in the summary
tree remains $k$.

Let $s_0$ denote the sum of the sizes of all the children of $v$ in $other_v-\{v_j\}$.
(Here ``$other_v$'' refers to $other_v$ before the change.)
Then $W$ times the increase in entropy in going from $S$ to $S'$ is given by
\[
I =  (s_0 + s_{v_i}) \lg \frac{W}{s_0 + s_{v_i}} + s_{v_j} \lg
\frac{W}{s_{v_j}} - (s_0 + s_{v_j}) \lg \frac{W}{s_0 + s_{v_j}}
- s_{v_i} \lg \frac{W}{s_{v_i}}.
\]
The derivative of this term with respect to $s_{v_i}$ is
$\lg
\frac{s_{v_i}}{s_0 + s_{v_i}} \le 0$.
As $i<j$, $s_{v_i} \le  s_{v_j}$, and thus $I$ is necessarily
nonnegative (for
it declines to 0 at $s_{v_i}=s_{v_j}$);
consequently, there is a nonnegative increase in entropy,
and hence $S'$ is also a maximum entropy summary tree.
Furthermore, if $d$ is the depth of $other_v$ in $S$,
then $\Phi(S')-\Phi(S)\le -(j-i)M^{n-d}<0$,
which contradicts the assumption that $S$ is a maximum
entropy summary tree of minimum $\Phi(S)$. \qed

\begin{lemma}
\label{lem:near-prefix}
Let $v$ be a node in $T$ such that $other_v$ exists in $S$.
If $v_i \notin other_v$ and ${v_{i+1}} \in other_v$, then
${v_j}
\notin other_v$ for all $j > i+1$.
\end{lemma}
\proof
Suppose, for a contradiction, that $v_j \in other_v$, for some $j > i+1$.

By Lemma~\ref{lem:subtrees-as-single-nodes},
$T_{v_i}$ is represented by two or more nodes in $S$.  Hence
$\{v_i\}$ appears as a node in the summary tree, and $\{v_i\}$  has one or more
children in $S$.  In $S$, let
$x$ be a descendant of $\{v_i\}$ of maximum depth in $S$.
Node $x$ is a proper descendant of $\{v_i\}$.

We will show now that combining node $x$ with another node
in a specified way yields a summary tree of $T_{v_i}$ with one
fewer node and having entropy at most $s_{v_i}$ smaller.
Node $x$ is not $\{v_i\}$.
Let $y$ be $x$'s parent in $S$.
Node $y=\{u\}$ for some node $u$ in $T$ (since every nonleaf in
a summary tree represents a single node of $T$).
There are four cases to analyze, but before turning to them,
we state the following simple
lemma  which we will need; it can be proven by calculus.
\begin{lemma}\label{easy}
If $a,b\ge 0$, $-a\lg a -b\lg b+(a+b)\lg (a+b) \le a+b.$
\end{lemma}
Let $s_x$, for a node $x$ in summary tree $S$,
denote the sum of the weights of all the nodes of $T$ represented
by $x$.  (For a node of the form $other_v$, we mean the sum of
the {\em sizes} of all the children of $v$ in $other_v$, or equivalently, the sum of
the weights of all their descendants.) % of nodes in $other_v$.)

Now we begin the case analysis.  Let $d$ be the depth in $S$ of node $\{v_i\}$.

\begin{enumerate}
\item
$y$'s only child in $S$ is $x$.

We combine nodes $x$ and $y = \{u\}$ into a node $z$
representing $T_{u}$.
Recall that $w_u$ denotes $u$'s weight. Then
$W$ times the entropy decrease equals
\begin{align*}
s_x\lg (W/s_x) & +w_u\lg (W/w_u) -(s_x+w_u)\lg (W/(s_x+w_u))\\
 =  &  ~~~ -s_x\lg s_x -w_u\lg w_u +(s_x+w_u)\lg (s_x+w_u)\\
 \le  & ~~~ s_x+w_u ~~~\text{(by Lemma~\ref{easy})}
~~~ =  ~ s_z ~ \le ~ s_{v_i}.
\end{align*}

This change leaves $\Phi$ unchanged.

\item
$x$ has a sibling in $S$ and $other_u$ does not exist.

Hence $x$ is either $\{\alpha\}$ or $T_{\alpha}$ for some node
$\alpha\in T$.

We create a new $other_u$ node by combining
$x$ with an arbitrary sibling $x'$ of $x$.  Because
$x$ is of maximum depth in $S$, $x'$ is either of the form
$\{\beta\}$ (node $\beta$ in $T$ has no children) or $T_{\beta}$, for some $\beta$ in
$T$.
The resulting entropy decrease equals
\begin{align*}
 s_x\lg (W/s_x) & +s_{x'}\lg (W/s_{x'})-(s_x+s_{x'}) \lg (W/(s_x+s_{x'})) \\
 = & ~~~ -s_x\lg s_x-s_{x'}\lg s_{x'}+(s_x+s_{x'})\lg (s_x+s_{x'})\\
\le  & ~~~ s_x+s_{x'} ~~~\text{(by Lemma~\ref{easy})}
~~~ \le ~ s_{v_i}.
\end{align*}

This change can increase $\Phi$ by at most $2n\cdot
M^{n-(d+1)}$,
because the depth of the new $other_u$ node is at least $d+1$.

\item
$x$ has a sibling in $S$ and $\{x\}=other_u$.

We choose an arbitrary sibling $x'$ of $x$ and add it to
$other_u$.  The entropy calculation is the same as for Case 2.
This change can increase $\Phi$ by at most $n\cdot
M^{n-(d+1)}$, where $d$ is the depth of $\{v_i\}$ in $S$.

\item
$x$ has a sibling in $S$, $other_u$ exists, and and $\{x\} \neq other_u$.

We add $x$ to $other_u$.  Let $x'$ be the node $other_u$.
The calculations are exactly the same as in Case 3.
\end{enumerate}

In all four cases, the decrease in entropy is at most $s_{v_i}$
and the increase in $\Phi$ is at most $2nM^{n-d-1}$.

Now we show how to generate a new maximum entropy summary tree
$S'$. To get $S'$,
combine $x$ as above with either its parent or a sibling,
thereby decreasing the number of summary tree nodes by one, and
then split off $v_{i+1}$ from $other_v$ and create a
node to represent $T_{v_{i+1}}$, thereby increasing the number
of summary tree nodes back to $k$.
Now, let $s_0$ denote the sum of the sizes of all the children of $v$ in
$other_v-\{v_{i+1},v_j\}$.
$W$ times the increase in entropy from this
two-part change to $S$ is at least
\begin{align*}
& \left [(s_0 + s_{v_j}) \lg \frac{1}{s_0 + s_{v_j}} +
s_{v_{i+1}} \lg \frac{1}{s_{v_{i+1}}} - (s_0 +
s_{v_{i+1}} + s_{v_j}) \lg \frac{1}{s_0 + s_{v_{i+1}} +
s_{v_j}}\right ]-s_{v_i} \\
&= (s_0 + s_{v_j}) \lg \frac{s_0 + s_{v_{i+1}} + s_{v_j}}{s_0 +
s_{v_j}} + s_{v_{i+1}} \lg \frac{s_0 + s_{v_{i+1}} +
s_{v_j}}{s_{v_{i+1}}} - s_{v_i}
\ge s_{v_{i+1}} - s_{v_i} \ge 0.
\end{align*}
(The first inequality
follows because $s_{v_j} \ge s_{v_{i+1}}$,
which implies that $(s_0+s_{v_{i+1}}+s_{v_j})/s_{v_{i+1}}\ge
2$.)
But this is a nonnegative increase in entropy, proving that $S'$
is a maximum entropy summary tree.

Splitting off $v_{i+1}$ from $other_v$ decreases $\Phi$ by at
least $M^{n-d}$, because the depth of the $other_v$ node equals
the depth of node $v_i$, which is $d$.
Hence the total $\Delta \Phi $ is at most $
-M^{n-d}+2n\cdot M^{n-d-1}=-M^{n-d} (1-2n/M)<0$,
a contradiction to the fact that $S$ is a maximum entropy
summary tree of minimum $\Phi$. \qed

This completes the proof of Theorem \ref{structural}. \qed

\begin{theorem}\label{othersize}
For all $v$, if $other_v$ exists, then $|other_v|\ge d_v-K+2$.
\end{theorem}
\proof
Each child of $v$ %  which is
not in $other_v$ contributes at least one node
to the final summary tree, which has order $k\le K$, and hence
the number of children not in $other_v$ cannot exceed $K-2$
(for % recall that
one node is needed to represent $\{v\}$).  \qed

\section{The Exact Algorithm}
\label{sec:exact}

Relabel the nodes as $1,2,...,n$, with the root being node 1,
the nodes at depth $d$ getting consecutive labels, and the
children of a node being labeled with increasing consecutive labels in nondecreasing size order.
(This can be done by processing the nodes in nondecreasing order
by depth, with
all the children of node $v$ processed consecutively in
nondecreasing order by size.)
This relabeling costs $O(n\log  n)$ time,%
\footnote{In fact,
the relative order, at node $v$, of its $d_v - K +1$
smallest-sized children
does not matter since they must all be included in $other_v$. This allows us to perform just a partial
sort at each node, in which the
$d_v - K +1$ smallest-size children are identified by selection
and then the remaining at most $K-1$ children are sorted.
This improves the $O(n\log n)$ term to $O(n\log K)$
which is dominated by $O(nK)$.}
because $\sum_v (d_v \log d_v)\le \sum_v (d_v\log n)\le n \log n$.

The description and the implementation of the algorithm are simplified if we compute
what we call the``pseudo-entropy,''
of summary trees for $T_{v}$ rather than their entropy.
The {\em pseudo-entropy p-ent$(S_v)$} of a tree $S_v$ with nodes of weights $W_1, W_2, \ldots, W_k$ is simply
$-\sum p_i \log p_i$, where $p_i = W_i /W$ and $W$ is the weight of $T$ (and \emph{not} of $T_{v}$).
Clearly, if $S_v$ is part of a summary tree $S$ for $T$, then $S_v$ contributes $-\sum p_i \log p_i$
to the entropy of $S$.
Let ent$(S_v)$ denote the entropy of tree $S_v$.
Then
\begin{eqnarray*}
\text{ent}(S_v) & = & - \sum_i \frac{W_i} {W_v} \log \frac {W_i} {W_v} ~ = ~
- \left[ \frac {W}{W_v} \sum_i \frac{W_i}{W} \log \frac{W_i}{W} + \sum_i {W_i}{W} \log \frac{W}{W_v} \right]\\
& = & - \frac {W}{W_v} \text{p-ent}(S_v) - \log \frac{W}{W_v}.
\end{eqnarray*}
Thus the same tree optimizes the entropy and the pseudo-entropy.
\hide{
Note that the proofs of Lemmas~\ref{lem:subtrees-as-single-nodes} and~\ref{lem:near-prefix}
apply unchanged to the pseudo-entropy and consequently the constraints on the nodes $other_v$
carry over when optimizing pseudo-entropy.
Henceforth, when we refer to an optimal summary tree, we intend a tree
optimizing the pseudo-entropy.
In fact, these trees also optimize the entropy,
a fact whose verification we
leave to the interested reader as it not needed to demonstrate the
of our algorithm.}

We will be using a dynamic programming algorithm.
To simplify the presentation we will only describe how to compute
the maximum pseudo-entropy for a $k$-node summary tree for $T_{v}$, for each node $v$
and for all $k$, $1 \le k \le \min\{K, n_{v}\}$.

The algorithm will first seek to find the value of
the pseudo-entropy for optimal $k$-node summary trees
when $other_v$ is restricted to being a prefix set,
and then when $other_v$ is restricted to being a near-prefix set
containing $v_j$ as its non-prefix element,
for each possible $v_j$ in turn, i.e.,\ for
$\max\{3,d_v - K + 3\} \le j \le d_v$.
Thus the algorithm will consider
$\min\{d_v-1, K-1\}$
$\min\{d_v, K-1\}$
classes of candidate $other_v$ sets.

To describe the algorithm it will be helpful to introduce the notion of
a summary forest.  A $k$-node {\em summary forest for $T_{v}$} is a $(k+1)$-node summary
tree for $T_{v}$ from which $v$ has been excised
(leaving a forest).
We will also call this a summary forest for $T_{v_1}, T_{v_2}, \ldots, T_{v_{d_v}}$.
A summary forest for $T_{v_1}, T_{v_2}, \ldots, T_{v_l}$ is defined analogously, for $1 \le l \le d_v$.

To find the pseudo-entropy-optimal $k$-node summary trees for $T_{v}$, for $1\le k \le K$, we first find
the pseudo-entropy of
optimal $k$-node summary forests for $T_{v_1},T_{v_2}, \ldots, T_{v_l}$, for $\max\{1,d_v - K+2\} \le l \le d_v$.
The optimal $k$-node summary trees for $T_{v}$ are then obtained by
attaching $\{v\}$
as a root node to the trees in the optimal $(k-1)$-node summary forests for $T_{v_1},T_{v_2}, \ldots,
T_{v_{d_v}}$.

Now we explain how to find these optimal summary forests.
In turn, we consider each of the up-to-$\max\{1, K-1\}$ possible classes of $other_v$ nodes:
the prefix $other_v$ nodes, and for each $j$ with
$\max\{3, d_v - K+3\} \le j \le d_v$,
the class of near-prefix $other_v$
nodes including $v_j$ as the non-prefix element.

First, we describe the handling of the candidate prefix $other_v$ nodes.
We start with optimal $k$-node summary trees for $T_{v_1}$, for $1 \le k \le K-1$.
Inductively, suppose that we have computed
(the entropy of) optimal $k$-node summary forests for $T_{v_1}, \dots, T_{v_l}$.
We find optimal $k$-node summary forests for $T_{v_1}, \dots, T_{v_l}, T_{v_{l+1}}$ as follows.
For $k=1$, the forest comprises a single $other_v$ node.
For each $k>1$,
we choose the highest entropy among the following options:
an optimal $h$-node summary forest for $T_{v_1}, \dots, T_{v_l}$ plus an
optimal $(k-h)$-node summary tree for $T_{v_{l+1}}$, for $1\le h < k$.

The correctness of this procedure is immediate: for $k=1$ clearly the only summary forest
is a one-node forest. For $k>1$, $T_{v_{l+1}}$ cannot be represented by the
$other_v$ node (since we are discussing the handling of the {\em
prefix} $other_v$ nodes) and so it must be represented by one
tree in the summary forest; this implies that $T_{v_1}, T_{v_2}, \ldots, T_{v_l}$ must also
be represented by one or more trees in the summary forest.
Of course, the representation of each of the parts must be optimal.
Our algorithm considers all possible ways of partitioning the nodes in the summary forest
among these two parts; consequently it finds an optimal forest.

The process when $v_j$ is the non-prefix node in $other_v$ is
essentially identical.
There are two changes: (i) $other_v$ is initialized to contain $T_{v_j}$ (rather than being the empty set) and
(ii) the incremental sweep skips tree $T_{v_j}$.
The correctness argument is as in the previous paragraph.

Finally, to obtain optimal $k$-node summary forests for $T_{v_1},T_{v_2}, \ldots, T_{v_{d_v}}$ one simply
takes the best among the $k$-node forests computed for the different classes of candidate $other_v$ nodes.
Again, correctness is immediate.

\begin{theorem}\label{time}
The running time of the algorithm is $O(K^2 n +n\log n)$.
\end{theorem}
\noindent
{\sc Note}.   Our time bound is $O(K^2n+n \log n)$ to build $K$ maximum-entropy
summary trees, or $O(Kn+(n\log n)/K)$ amortized time for each.  There is an obvious
lower bound of $\Omega (n+K^2)$ to build all $K$ trees, since one has to read an
$n$-node tree and produce trees having $1,2,3,\ldots,K$ nodes.
Hence there cannot be a $O(n)$-time algorithm that generates all
$K$ trees,
since it would violate the lower bound when $K$ is
$\omega(\sqrt{n})$.
Of course, conceivably there is a linear-time algorithm
to build a maximum-entropy $k$-node
summary tree for a single value of $k$.
\proof
The running time is the sum of three terms: \\
(1) $O(n \log n)$,  for sorting the children of all nodes by size.
\\
(2) $O(K n)$ for initializations.  In fact, the initializations
for node $v$ take time $O(K\cdot \min\{d_v,K-1\} )$, which is $O(Kn)$ time in total.
\\
(3) For each node $v$, the cost of processing node $v_l$
when processing each of the classes of candidate $other_v$ nodes.
Let $\langle v_a,v_{a+1},\ldots,v_{v_d}\rangle $ be the sequence of nodes processed
when considering the candidate prefix $other_v$ sets
(nodes $v_1, \ldots, v_{a-1}$ are the nodes guaranteed to be in $other_v$).
When processing the near-prefix candidate $other_v$ sets with non-prefix element $v_j$,
the same sequence will be processed except that $v_j$ will be omitted.
For the class of prefix candidate sets, the cost for processing $v_{l+1}$, for $a \le l < v_d$, is
$\min\{K-1, n_{v_a} + n_{v_{a+1}} + \cdots + n_{v_l}\}\cdot \min\{K-1, n_{v_{l+1}}\}
\le \min\{K-1, n_{v_1} + n_{v_2} + \cdots + n_{v_l}\}\cdot \min\{K-1, n_{v_{l+1}}\}$,
for we are seeking $k$-node summary forests for $1 \le k \le K-1$, and the number of nodes
in a summary tree cannot be more than the number of nodes available in the relevant subtrees of $T$.
The same bound applies for each of the remaining classes of candidate $other_v$ sets and
there are at most $K-1$ of these classes.
Since the number of child nodes being processed when computing at node $v$ is $d_v -a +1 \le d_v$,
the obvious upper bound here
is $O(K^3\cdot d_v)$.
Summed over all $v$, this totals $O(K^3 \cdot n)$.
However, Corollary~\ref{arbperm} below shows that
$\sum_{\text{non-leaf } v}\sum_l \min\{ n_{v_1}+ n_{v_2}+\cdots+
n_{v_l},K\} \cdot
\min\{n_{{v_{l+1}}},K\}\le 2Kn$, giving an overall
time of $O(n\log n+K^2 n)$.
\qed

\section{A Lemma For Running Time Analysis}
\label{sec:key-lemma}

In this section we state  a lemma underlying the running time
analysis of both the greedy algorithm and the exact algorithm.
Let $n$ be a positive integer and let $T$ be a rooted, $n$-node tree, and
\emph{for this section only},
let $v_1,v_2,...,v_{d_v}$ be $v$'s children in {\em any} order.

\begin{definition}
\label{def:cost-v}
Relative to $T$, let $cost(v)$ be defined for all $v\in T$ as
follows.
If $v$ is a leaf, $cost(v)=0$.
If $v$ is not a leaf,
$cost(v)=[\sum_{i=1}^{d_v} cost(v_i)]+
[\sum_{i=1}^{d_v-1} \min\{n_{v_1}+n_{v_2}+\cdots + n_{v_i},K\}\cdot \min\{n_{{v_{i+1}}},K\}].$
\end{definition}

\begin{lemma} \label{greedythm}
For all $v$, $cost(v)\le n_{v}^2$ if $n_{v}\le K$, and
$cost(v)\le 2Kn_{v}-K^2$, if $n_{v}> K$.
\end{lemma}
\proof
We prove the lemma by
induction on the height  of $v$ (i.e., the maximum length of a path from $v$
down to a leaf).

\noindent
Basis.  If $v$ has height 0, i.e., $v$ is a leaf, then
$cost(v)=0$, whereas $n_{v}=1\le K$ and indeed $cost(v)=0\le 1= n_{v}^2$.

\noindent
Inductive step. Let $h\ge 0$ and assume that the statement is true for all nodes
of height at most $h$.

To simplify the notation, we will use $n_i$ to denote $n_{v_i}$
and $d$ to denote $d_v$ from now on.

Let $v$ be a node of height $h+1$;  then
$v$'s children have height at most $h$.  Therefore, by
induction, if $v$'s children are $v_1,v_2,...,v_{d}$, then
$cost(v_i)\le n_{i}^2$ if $n_{i}\le K$, and
$cost(v_i)\le 2Kn_{i}-K^2$, if $n_{i}>K$.

Let $Cost_j =\sum_{1 \le i \le j cost v_i}  +[\sum_{i=1}^{j-1} \min\{n_{1}+n_{2}+\cdots +
n_{i},K\}\cdot \min\{n_{{i+1}},K\}]$.
We will show by induction on $j$ that if $\sum_{1 \le i \le j} n_i \le K$,
then $Cost_j \le (\sum _{1 \le i \le j} n_i)^2$,
and otherwise $Cost_j \le [2K\cdot \sum _{1 \le i \le j} n_i] - K^2$,
from which the result in the lemma is immediate.

Let $t_j = \sum_{i=1}^j n_i$ and let $u_j = n_{j+1}$.
There are five cases to consider.

\smallskip
\noindent
i. $t_j + u_j \le K$. Then
\[
Cost_j ~\le~ t_j^2 + u_j^2 + t_j u_j ~\le~ (t_j + u_j)^2 ~=~ ( \sum_{i=1}^{j+1} n_i)^2.
\]

\smallskip
\noindent
ii. $t_j > K$ and $u_j  \le K$. Then
\begin{eqnarray*}
Cost_j & \le & (2K t_j -K^2) + u_j^2 + Ku_j ~\le~ 2K t_j -K^2 + 2Ku_j \\
& \le & 2K(t_j + u_j) - K^2 ~=~ [2K \sum_{i=1}^{j+1} n_i] -K^2.
\end{eqnarray*}

\smallskip
\noindent
iii. $t_j \le K$ and $u_j > K$.
\\
This has essentially the same analysis as Case ii.

\smallskip
\noindent
iv. $t_j > K$ and $u_j >K$. Then
\begin{eqnarray*}
Cost_j & \le & (2K t_j -K^2) + (2K u_j - K^2) + K^2 \\
& \le & 2K(t_j + u_j) - K^2 ~=~ [2K \sum_{i=1}^{j+1} n_i] -K^2.
\end{eqnarray*}

\smallskip
\noindent
v. $t_j \le K$, $t_j + u_j > K$, and $u_j \le K$.\\
Let $\Delta = t_j + u_j - K$. Then
\begin{eqnarray*}
Cost_j & \le & t_j^2 + u_j^2 + t_j u_j = t_j^2 + (K + \Delta - t_j)^2 + t_j(K + \Delta - t_j) \\
& \le & t_j^2 - Kt_j - \Delta t_j + K^2 + \Delta^2 + 2K \Delta \\
&=& 2(\Delta + K) K - K^2 + \Delta^2  - \Delta t_j + t_j^2 - Kt_j\\
& \le & 2(\Delta + K) K - K^2 ~\le ~ 2K\sum_{i=1}^{j+1} n_i - K^2.
\end{eqnarray*}
The next-to-last inequality follows because
 $\Delta \le t_j$
and
$t_j \le K$.
\qed

\begin{corollary} \label{arbperm}
For $K \ge 1$,
$\sum_{\text{non-leaf}~v} [\sum_{i=1}^{d_v-1}
\min\{n_{v_1}+n_{v_2}+\cdots +
n_{v_i},K\}\cdot \min\{n_{{v_{i+1}}},K\}]\le 2Kn.$
\end{corollary}
\proof
Summing over all non-leaf nodes in Definition~\ref{def:cost-v} yields that
the term we are bounding equals $cost(r)$, where $r$ is the root of the tree;
the result now follows from Lemma~\ref{greedythm}.
\qed

\section{Greedy Algorithm}
\label{sec:greedy}

The greedy algorithm proposed in \cite{KS} is precisely
the algorithm proposed herein for the exact solution but with the {\em other} sets
restricted to being prefix sets.
\remove{
Algorithm \ref{opt}, when $P_v$ is replaced by the singleton
$\{\langle 1,2,...,d_v\rangle\}$.  (Recall that the children of node
$v$ were already sorted into nondecreasing order by size.)
}
In~\cite{KS} Greedy was shown to run in time $O(K^2 n+n
\log n)$.  Here, we shave off a factor of $K$ from the first
term.

\begin{corollary}\label{greedytime} (of Lemma~\ref{greedythm}).
The time needed by the greedy algorithm to generate
summary trees of orders $k=1,2,\ldots,K$ is $O(Kn+n\log n)$.
\end{corollary}
\proof
Aside from
initializations (which take time $O(Kn)$) and sorting (which takes time $O(n \log n)$), the time needed by the greedy
algorithm is $O$ of
$$\sum_v [\sum_{i=1}^{d_v-1} \min\{n_{v_1}+n_{v_2}+\cdots +
n_{v_i},K\}\cdot \min\{n_{v_{i+1}},K\}],$$ which by Corollary
\ref{arbperm} is at most $2Kn$, giving an overall bound of $O(Kn+n\log n)$.
\qed

Again, one can reduce the $O(n \log n)$ term to $O(n \log K)$, giving an overall run time of $O(nK)$.
Here we rely on the fact that the greedy algorithm will
necessarily put $v_1,v_2,...,v_{d_v-K}$ into $other_v$.  To save
time, we can modify Greedy so as to put those nodes into
$other_v$ and only individually process children $v_{d_v-K+1},
v_{d_v-K+2}, ..., v_{d_v}$.   We can find the $d_v-K$ children
of $v$ of least size via a selection (not sorting) algorithm and
then sort only the remaining $K$ children.  This makes the total
sorting time over all nodes $v$ $O$ of
$\sum_v \min\{d_v,K\}\log (\min\{d_v,K\})\le \sum_v
\min\{d_v,K\}\log K\le (\log K) \sum_v d_v=n\log K$.

\section{Improved Approximation Algorithm}
\label{sec:approx}

In this section we describe an algorithm that computes an approximately entropy-optimal $k$-node summary tree.
Our algorithm relies on the following outline from \cite{KS}:
\begin{enumerate}
\item
One can rescale the weights in a tree
to make them sum up to any positive integral value $W_0$, while
leaving the entropy of any summary tree unchanged.  (This is
obvious.)
\item
One can use algorithmic discrepancy theory to round each
resulting real node weight $w_v$ to value $w'_v$ equal to either
$\lfloor w_v \rfloor$ or $1+\lfloor w_v \rfloor$ such that for
each node $v\in T$,
$|\sum_{u \in T_{v}} w'_u
-\sum_{u \in T_{v}} w_u|\le 1$ {\em for all $v$ simultaneously},
without changing the overall sum.
\item
Using Naudts's theorem \cite{naudts}
that almost identical probability distributions have almost
identical entropy, one can prove, for some
integer $W_0$ which is $O((K/\epsilon) \log (K/\epsilon))$,
that if one finds
a maximum entropy summary tree $T^*$
for the modified weights $(w'_v)$,
then $T^*$ has entropy (measured according to the {\em original} weights
$w_v$)
at most $\epsilon$ less than that of the truly maximum
entropy summary tree.
\end{enumerate}

Suppose that the weights on $T$ are integral and sum to $W_0$.
Clearly the number of nodes of positive weight cannot exceed
$W_0$; however, the 0-weight nodes could far outnumber the
positive-weight nodes.  Indeed, that is exactly what happens if
$n \gg W_0$.

Our algorithm exploits the fact that little processing is needed
for most of the 0-weight nodes.
In fact, we will need to compute summary trees for only
the non-zero weight nodes and for at most
$2(W_0 - 1)$ 0-weight nodes.

The algorithm works with a tree $T'$, a reduced version of $T$ in which
some 0-weight nodes have been removed.
The following notation will be helpful.
$F_T(v,k)$ denotes the maximum pseudo-entropy of a
$k$-node summary tree of $T_{v}$, where $T_{v}$ is a subtree of tree $T$;
similarly, $F_{T'}(v,k)$ denotes the maximum pseudo-entropy of a
$k$-node summary tree of $T'_v$, where $T'_v$ is a subtree of tree $T'$.

$T'$ is obtained from $T$ as follows: for each positively-sized node $v$ in $T$,
if $v$ has one or more size-0 children, remove them and their
descendants and replace them all by
a single 0-weight child.
Clearly optimal summary trees in $T'$ form optimal summary trees in $T$
(for the only difference in summarizing $T$ is that we could add
0-weight nodes no longer present in $T'$, and these would contribute
0 to the entropy).
Note that if $v$ is a 0-weight non-leaf node in $T'$ then it
must have non-zero size (assuming $T$ has at least one positive-weight node).
The following result is immediate.

\begin{lemma}
Let $T$ have $n$ nodes and $T'$ have $n'$ nodes.
Let $v$ be a node in $T'$ with
$n(v)$ descendants in $T$ and $n'(v)$ descendants in $T'$.
Then $F_T(v,k)=F_{T'} (v,k)$ for $1\le k\le n'(v)$.
For $n'(v)+1\le k\le n$, $F_T(v,k)=F_{T'} (v,n'(v))$.
\end{lemma}
Note that $F_{T'} (v,n'(v))$ is attained by a partition of the
set of $v$'s children in
$T'$ into singletons.

(Now of course we {\em have} changed the problem, since
$T'$ might have fewer than $K$ nodes.  However,
if this happens, then optimal summary trees of $T$ having more
than $|T'|$ nodes have no more entropy than optimal summary trees of $T$
having exactly $|T'|$ nodes.)

Even after the reduction it may be the case that $|T'| \gg W_0$,
for $T'$ might still contain long paths of 0-weight nodes in which
each node has only one positively-sized child.
However, the following lemmas show that they add little to the cost of
computing optimal summary trees.

\begin{lemma}
\label{lem:0-wt-1-subtree}
Let $v$ be a 0-weight node in $T'$ with a single child $u$.
Then for $2\le k \le |T'_v|$,
$F_{T'}(v,k+1) = F_{T'}(u,k)$;
also $F_{T'}(v,1) = F_{T'}(u,1)$.
\end{lemma}
\begin{proof}
For $k\ge 2$, the $(k+1)$-node summary tree for $T'_v$ adds a
zero-weight node $\{v\}$ to the
$k$-node summary tree for $T'_u$.
For $k=1$ both trees have a single node of weight $w_u$.
\end{proof}

\begin{lemma}
\label{lem:0-wt-2-subtree}
Let $v$ be a 0-weight node in $T'$ with exactly two children, a 0-weight leaf $v_1$
and a child $u$ of positive size.
Then for $3\le k \le |T'_v|$,
$F_{T'}(v,k+2) = F_{T'}(u,k)$;
also $F_{T'}(v,2) = F_{T'}(v,1) = F_{T'}(u,1)$.
\end{lemma}
The proof of this lemma is essentially the same as that of
Lemma~\ref{lem:0-wt-1-subtree}.
The following corollary is immediate.

\begin{corollary}
\label{cor:path-red}
Let $v_1, v_2, \ldots, v_l$, for $l>1$, be a descending path of 0-weight nodes in $T'$
such that each $v_i$, $1 \le i \le l$ either has one child, or
has exactly two children
one of which is a 0-weight leaf.
Further suppose that $l'$ of these nodes are in the second category.
Node $v_l$ must have a child of positive size (as otherwise
$v_1\ne v_l$ would be a size-0 non-leaf).
Let $u$ be the child of $v_l$ of positive size.
Then for $1\le k \le |T'_{v_1}| - (l+l')$,
$F_{T'}(v_1,k+l+l') = F_{T'}(u,k)$;
and for $j \le l+l'$,
$F_{T'}(v_1,j) = F_{T'}(u,1)$.
\end{corollary}

This corollary implies that given the entropies
of optimal entropy summary trees
at a node $u$ at the bottom of a maximal path of 0-weight nodes, one
can obtain the entropies
of the optimal entropy summary trees at node
$v_1$ at the top of the path in time $O(K)$.

At the remaining nodes in $T'$
we perform the same computation
as in the exact algorithm. As we can show, there are $O(W_0)$ such
nodes, which leads to the following running time bound.

\begin{theorem}
\label{thm:approx-runtime}
The approximation algorithm to obtain a summary tree that has entropy within an additive
$\eps$ of that of the optimal summary  tree runs in time $O(n + W_0 \cdot K^3)$,
where $W_0 = O((K/\eps) \log (K/\eps))$.
\end{theorem}
\proof
We begin by bounding the numbers of nodes of various types.
Clearly, there are at most $W_0$ non-zero weight nodes.
Thus there are at most $(W_0 -1)$  0-weight nodes
with two or more non-zero weight subtrees.
All other 0-weight nodes are either leaf nodes or lie
on maximal paths of 0-weight nodes with one non-zero
weight subtree.
Further, there can be at most $W_0 - 1$ such maximal paths.

The naive bound on the cost of the computation at a node $v$ in the exact algorithm
is $O(K^3 \cdot d_v)$ (see the proof of Theorem~\ref{time}),
and this applies to the non-zero weight nodes and the 0-weight nodes
with two or more non-zero weight subtrees, giving a cost of $O(K^3 \cdot W_0)$
in total.

The cost of processing the maximal paths of 0 weight is $O(W_0\cdot K)$.

Finally, recall that we need to sort the children in non-decreasing order by size
for each parent node with non-zero weight
or with two or more non-zero weight subtrees.
We implement this by means of  a radix sort on
the pairs $(\text{parent}, W_i)$, over these parent nodes
(recall that $W_i$ is the size of the $i$th child).
There are $O(W_0)$ such pairs, with indices in the ranges $n$ and $W_0$
respectively, yielding a running time of $O(n + W_0)$ for the radix sort.
\qed

In contrast to the $O(nK^2 + n\log n)$ bound for the exact algorithm, here the
sophisticated analysis of Lemma~\ref{greedythm}
cannot be applied.
The reason is that Lemma~\ref{greedythm}
assumes that a tree of $r$ nodes yields at most $r$ optimal
summary trees each having a distinct entropy (so if the $k$-node
and $(k+1)$-node optimal trees have the same entropy only one of
them is counted).
However, the same claim fails to hold for trees having
$r$ {\em positive-weight} nodes (and a total number of nodes
potentially vastly exceeding $r$), as would be needed in order to apply
Lemma~\ref{greedythm} here.

\subsubsection*{Acknowledgments.}
We thank the referees for their helpful suggestions.

%\vspace*{-0.1in}

% \pagebreak


\begin{thebibliography}{MMM}
\bibitem{beerman}
Dale Beermann, Tamara Munzner, and Greg Humphreys,
        ``Scalable, Robust Visualization of Very Large Trees,''
        {\em Proc. EuroVis}, 2005, 37--44.
\bibitem{card}
S. K. Card and D. Nation, ``Degree-Of-Interest Trees: A Component of an Attention-Reactive User Interface,''
{\em Proceedings of the Working Conference on Advanced Visual Interfaces}, 2002, 231--245.
%\bibitem{yannakakis}
%Kousha Etessami, Alistair Stewart, Mihalis Yannakakis,
%``A Note On the Complexity of Comparing Succinctly Represented
%Integers, With An Application to Maximum Probability Parsing,"
%CoRR abs/1304.5429 (2013) or ArXiv:1302.6411 (2013).
\bibitem{heer}
Jeffrey Heer and Stuart K. Card,
``DOI Trees Revisited: Scalable, Space-Constrained Visualization of Hierarchical Data,"
in {\em Advanced Visual Interfaces}, 2004,
\url{//vis.stanford.edu/papers/doitrees-revisited}, 421--424.
\bibitem{KS}
Howard Karloff and Ken Shirley, ``Maximum Entropy Summary Trees," Eurovis
2013; %(best paper honorable mention).
\url{www2.research.att.com/~kshirley/KarloffShirleyWebsite.pdf}.
\bibitem{lamping}
John Lamping, Ramana Rao, and Peter Pirolli,
        ``A Focus+Context Technique Based on Hyperbolic Geometry For Visualizing Large Hierarchies,''
        {\em Proceedings of the SIGCHI Conference on Human Factors in Computing Systems},
        \url{//dx.doi.org/10.1145/223904.223956}, 1995,
401--408.
\bibitem{vL}
T. von Landesberger, A. Kuijper, T. Schreck, J. Kohlhammer, J. J. van Wijk, J.-D.
Fekete and D. W. Fellner,
        ``Visual Analysis of Large Graphs: State-of-the-Art and Future Research Challenges,"
        {\em Computer Graphics Forum}, vol. 30, no. 6, 2011, 1719-1749.
\bibitem{munzner}
T. Munzner, R. Guimbretiere, S. Tasiran, L. Zhang, and Y. Zhou,
        ``TreeJuxtaposer: Scalable tree comparison using Focus+Context with guaranteed visibility,''
        {\em ACM Transactions on Graphics}, vol. 22, no. 3, 2003, 453--462.
\bibitem{naudts}
J. Naudts, ``Continuity of a Class of Entropies and Relative
Entropies," {\em Reviews in Mathematical Physics} 16, 6 (2004),
809-822.
\bibitem{shneiderman}
Ben Shneiderman, ``The Eyes Have It: A Task by Data Type Taxonomy for Information Visualization,''
        {\em Proceedings of the IEEE Symposium on Visual Languages}, 1996, 336--343.
\end{thebibliography}
\end{document}